\documentclass[10pt]{article}

\usepackage{graphicx} 
\usepackage[T1]{fontenc}
\usepackage[utf8x]{inputenc}
\usepackage{amsmath, amssymb, amsthm}
\usepackage{csquotes}
\usepackage[english]{babel}
\usepackage{eurosym}
\usepackage{fullpage}
\usepackage{palatino}
\usepackage{tikz}
\usepackage{verbatim}
\usepackage{fancybox}

\usepackage{subfig}
\usepackage{enumerate}
\usepackage{epsfig}
\usepackage[active]{srcltx}
\usepackage{amsfonts}
\usepackage{amsmath}
\usepackage[mathlines]{lineno}
\usepackage{xspace}
\usepackage{tikz}

\usetikzlibrary{arrows,shapes,snakes,automata,backgrounds,petri,patterns}

\newtheorem{lemma}{Lemma}

\newtheorem{theorem}{Theorem}
\newtheorem{claim}{Claim}

\newcounter{claimb}
    
\def\claimb{$$\vcenter\bgroup\advance\hsize by -8em\noindent
\refstepcounter{claimb}\ignorespaces\it}        
\makeatletter
\def\endclaimb{\rm\egroup\leqno(\theclaimb)$$\global\@ignoretrue}
\makeatother

\newenvironment{proofclaim}[1][]%
    {\noindent \emph{Proof.} {}{#1}{}}{\hfill
    $\Diamond$\vspace{1em}}

\begin{document}

\title{Planar graphs with $\Delta\geq 7$ and no triangle adjacent to a $C_4$\\are minimally edge- and total-choosable.\thanks{Work supported by the ANR Grant EGOS (2012-2015) 12 JS02 002 01}}%\thanks{This result was presented in Cycles\&Colorings'11}}

\author{Marthe Bonamy, Benjamin Lévêque, Alexandre Pinlou\thanks{Second affiliation: D\'epartement MIAp, Universit\'e Paul-Val\'ery, Montpellier 3}\\ \normalsize{LIRMM, Universit\'e Montpellier 2, CNRS}\\ \small{\{marthe.bonamy, benjamin.leveque, alexandre.pinlou\}@lirmm.fr}}

\maketitle

\begin{abstract}
For planar graphs, we consider the problems of \emph{list edge coloring} and \emph{list total coloring}. Edge coloring is the problem of coloring the edges while ensuring that two edges that are adjacent receive different colors. Total coloring is the problem of coloring the edges and the vertices while ensuring that two edges that are adjacent, two vertices that are adjacent, or a vertex and an edge that are incident receive different colors. In their list extensions, instead of having the same set of colors for the whole graph, every vertex or edge is assigned some set of colors and has to be colored from it. A graph is minimally edge or total choosable if it is list edge $\Delta$-colorable or list total $(\Delta+1)$-colorable, respectively, where $\Delta$ is the maximum degree in the graph.

It is already known that planar graphs with $\Delta\geq 8$ and no triangle adjacent to a $C_4$ are minimally edge and total choosable (Li Xu 2011), and that planar graphs with $\Delta\geq 7$ and no triangle sharing a vertex with a $C_4$ or no triangle adjacent to a $C_k$ ($\forall 3 \leq k \leq 6$) are minimally total colorable (Wang Wu 2011). We strengthen here these results and prove that planar graphs with $\Delta\geq 7$ and no triangle adjacent to a $C_4$ are minimally edge and total choosable.
\end{abstract}

\section{Introduction}

We consider simple, connected graphs. An \emph{edge $k$-coloring} of a graph $G$ is a coloring of the edges of $G$ with $k$ colors such that two edges that are adjacent receive distinct colors. We define $\chi'(G)$ as the smallest $k$ such that $G$ admits an edge $k$-coloring. An extension of the problem of edge coloring is the \emph{list edge coloring}, defined as follows. For any list assignment $L:E\rightarrow \mathcal{P}(\mathbb{N})$, a graph $G=(V,E)$ is edge $L$-colorable if there exists an edge coloring of $G$ such that the color of every edge $(u,v)\in E$ belongs to $L(u,v)$. A graph $G=(V,E)$ is said to be \emph{list edge $k$-colorable} (or \emph{edge $k$-choosable}) if $G$ is edge $L$-colorable for any list assignment $L$ such that $|L(u,v)|\geq k$ for any edge $(u,v)\in E$. We define $\chi'_\ell(G)$ as the smallest $k$ such that $G$ is edge $k$-choosable. A \emph{total $k$-coloring} of a graph $G=(V,E)$ is a coloring of its edges and vertices with $k$ colors such that two elements of $V \cup E$ that are adjacent or incident receive distinct colors. The definitions naturally extend to \emph{list total coloring}, and $\chi''_\ell(G)$.\\

Note that for any graph $G$, we have $\chi'_\ell(G)\geq\chi'(G)\geq\Delta(G)$ and $\chi''_\ell(G)\geq\chi''(G)\geq\Delta(G)+1$, where $\Delta(G)$ denotes the maximum degree of a vertex in $G$. The first parts of the two inequalities are both conjectured to be equalities (List Coloring Conjecture, and~\cite{bkw97}).

Sufficient conditions for all these inequalities to be actually equalities have been extensively studied. We say that a graph is minimally \emph{edge} or \emph{total} \emph{choosable} if it is edge $\Delta(G)$-choosable or total $(\Delta(G)+1)$-choosable, respectively. For the planar graph case, the best known result is that when $\Delta(G)\geq 12$, a planar graph $G$ is minimally edge- and total-choosable~\cite{bkw97}. For $k \geq 3$, a cycle of length $k$ (resp. at most $k$) is denoted $C_k$ (resp. $C_{k^-}$). Two cycles are said to be \emph{incident} if they share at least one vertex, and \emph{adjacent} if they share at least one edge. When adding restrictions on the cycles, it is known for example that planar graphs with $\Delta(G)\geq 7$ and no $C_4$~\cite{hlc06} or no two adjacent $C_{4^-}$~\cite{lmw13}, or $\Delta(G)\geq 8$ and no triangle adjacent to a $C_4$~\cite{lx11} are minimally edge- and total-choosable. Regarding total coloring only, it is known that planar graphs with $\Delta(G)\geq 7$ and no triangle incident to a $C_4$~\cite{ww11} or no triangle adjacent to a $C_k$ ($k \in \{3,4,5,6\}$)~\cite{ww11} are minimally total-choosable.

Here we strengthen these results by proving that:
\begin{theorem}\label{thm:main}
Every planar graph with $\Delta(G) \geq 7$ and no triangle adjacent to a $C_4$ satisfies $\chi'_\ell(G)=\Delta(G)$ and $\chi''_\ell(G)=\Delta(G)+1$.
\end{theorem}

In Sections~\ref{sect:meth} and \ref{sect:def}, we introduce the method and terminology. In Sections~\ref{sect:conf} and~\ref{sect:dis}, we prove in two steps Theorem~\ref{thm:main}, with the discharging methods described in Section~\ref{sect:def}. \newline

\section{Method}\label{sect:meth}

The discharging method was introduced in the beginning of the 20$^{th}$ century. It has been used to prove the celebrated Four Color Theorem (\cite{ah77} and \cite{ahk77}).

We prove Theorem~\ref{thm:main} using a discharging method, as follows. A graph is \emph{minimal} for a property if it satisfies this property but none of its proper subgraphs does. The first step is to set an integer $k \geq 7$ and consider a minimal counter-example $G$ (i.e. a graph $G$ such that $\Delta(G) \leq k$ and $\chi'_\ell(G) > k$ or $\chi''_\ell(G)>k+1$, whose every proper subgraph is $k$-edge-choosable and $(k+1)$-total-choosable), and prove it cannot contain some configurations. To that purpose, we assume by contradiction that $G$ contains one of the configurations. We consider a particular subgraph $H$ of $G$. For any list assignment $L$ on the edges of $G$, with $|L(e)|\geq k$ for every edge $e$, we $L$-edge-color $H$ by minimality. We show how to extend the $L$-edge-coloring of $H$ to $G$, a contradiction. We argue that, except in a well-specified case, the same proof works for $L$-total-coloring with any list assignment $L$ on the edges and vertices of $G$, with $|L(e)|\geq k+1$ and $|L(v)|\geq k+1$ for every edge $e$ and vertex $v$.

The second step is to prove that a connected planar graph on at least two vertices with $\Delta \leq k$ that does not contain any of these configurations nor a triangle adjacent to a $C_4$ does not satisfy Euler's Formula. To that purpose, we consider a planar embedding of the graph. We assign to each vertex its degree minus six as a weight, and to each face two times its degree minus six. We apply discharging rules to redistribute weights along the graph with conservation of the total weight. As some configurations are forbidden, we can prove that after application of the discharging rules, every vertex and every face has a non-negative final weight. This implies that $\sum_v(d(v)-6)+\sum_f(2d(f)-6) = 2\times |E(G)|- 6\times |V(G)| + 4 \times |E(G)| - 6 \times |F(G)| \geq 0$, a contradiction with Euler's Formula that $|E| - |V| - |F|=-2$. Hence a minimal counter-example cannot exist.

\section{Terminology}\label{sect:def}
Let $k\geq 7$.\\

In the figures, we draw in black a vertex that has no other neighbor than the ones already represented, in white a vertex that might have other neighbors than the ones represented. When there is a label inside a white vertex, it is an indication on the number of neighbors it has. The label '$i$' means "exactly $i$ neighbors", the label '$i^+$' (resp. '$i^-$') means that it has at least (resp. at most) $i$ neighbors. The same goes for faces. Note that the white vertices may coincide with other vertices.\\

Given a plane graph (ie a planar graph with its embedding) and a face $f=(u,v,w,x)$, we say $w$ is the vertex \emph{opposite} to $u$ in $f$. If there is a face $(t,u_1,v,u_2)$ with $d(t)=2$ and $d(v)=3$, we say a neighbor $w$ of $u_1$ is the \emph{$(v,u_1)$-support of $t$} if the sequence $(t,v,v_1,v_2,\ldots,v_{p-1},v_p=w)$ of consecutive neighbors of $u_1$ contains only vertices of degree $3$ except for $t$, and any two consecutive neighbors of the sequence are part of the boundary of a face of degree $4$ that contains $u_1$, while the edge $(u_1,w)$ belongs to a face of degree at least $5$, or to a face of degree $4$ with a vertex of degree at least $4$ opposite to $w$ (see Figure~\ref{fig:parent}). Given $t$, $v$ and $u_1$, at most one vertex can satisfy this property. Note that $v$ can itself be the $(v,u_1)$-support of $t$, and that it can even be also the $(v,u_2)$-support of $t$. Note that, by definition, if $w$ is the $(v,u_1)$-support of $t$, then the edge $(u_1,w)$ is incident, on one side, to either a face of degree at least $5$, or to a face of degree $4$ where the vertex opposite to $w$ is of degree $\geq 4$, and, on the other side, to a face of degree $4$ where the vertex opposite to $w$ is of degree $3$. Consequently, a vertex cannot be support more than twice, as a support vertex is of degree $3$. 

\captionsetup[subfloat]{labelformat=empty}
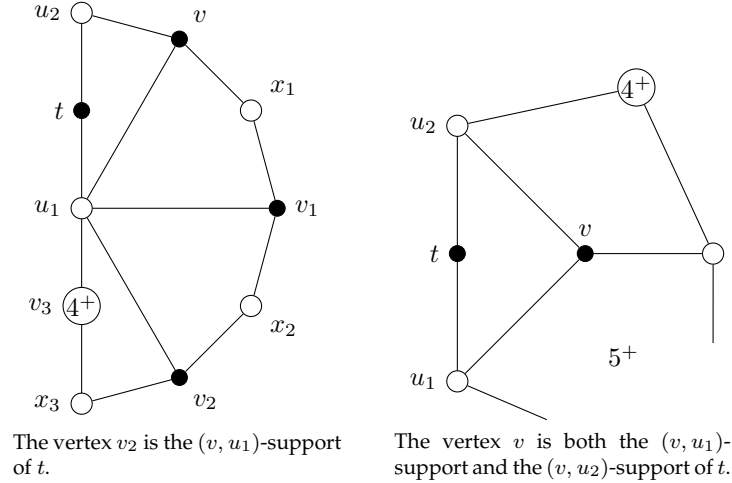
\begin{figure}[!h]
\centering
\subfloat[][The vertex $v_2$ is the $(v,u_1)$-support of $t$.]{
\centering
 \begin{tikzpicture}[scale=1.3]
    \tikzstyle{whitenode}=[draw,circle,fill=white,minimum size=8pt,inner sep=0pt]
    \tikzstyle{blacknode}=[draw,circle,fill=black,minimum size=6pt,inner sep=0pt]
\tikzstyle{tnode}=[draw,ellipse,fill=white,minimum size=8pt,inner sep=0pt]
\tikzstyle{texte} =[fill=white, text=black]

 \draw (2,-2) node[whitenode] (e4) [label=180:$x_3$] {}
-- ++(90:1cm) node[whitenode] (d2) [label=180:$v_3$] {$4^+$}
-- ++(90:1cm) node[whitenode] (c) [label=180:$u_1$] {}
-- ++(90:1cm) node[blacknode] (d1) [label=180:$t$] {}
-- ++(90:1cm) node[whitenode] (e1) [label=180:$u_2$] {};

\draw (3,1.732) node[blacknode] (f1) [label=60:$v$] {};
\draw (3.732,1) node[whitenode] (e2) [label=30:$x_1$] {};
\draw (4,0) node[blacknode] (g) [label=0:$v_1$] {};
\draw (3,-1.732) node[blacknode] (f3) [label=-60:$v_2$] {};
\draw (3.732,-1) node[whitenode] (e3) [label=-30:$x_2$] {};

\draw (e1) edge node {} (f1);
\draw (f1) edge node {} (e2);
\draw (e4) edge node {} (f3);
\draw (g) edge node {} (e2);
\draw (g) edge node {} (c);
\draw (g) edge node {} (e3);
\draw (f3) edge node {} (e3);
\draw (c) edge node {} (f1);
\draw (c) edge node {} (f3);
\end{tikzpicture}
\label{fig:p1}
}
\qquad
\subfloat[][The vertex $v$ is both the $(v,u_1)$-support and the $(v,u_2)$-support of $t$.]{
\centering
 \begin{tikzpicture}[scale=1.7]
    \tikzstyle{whitenode}=[draw,circle,fill=white,minimum size=8pt,inner sep=0pt]
    \tikzstyle{blacknode}=[draw,circle,fill=black,minimum size=6pt,inner sep=0pt]
\tikzstyle{tnode}=[draw,ellipse,fill=white,minimum size=8pt,inner sep=0pt]
\tikzstyle{texte} =[fill=white, text=black]

 \draw (0,0) node[whitenode] (c) [label=180:$u_1$] {}
-- ++(90:1cm) node[blacknode] (d1) [label=180:$t$] {}
-- ++(90:1cm) node[whitenode] (e1) [label=180:$u_2$] {};

\draw (1,1) node[blacknode] (f) [label=90:$v$] {}
-- ++(0:1cm) node[whitenode] (g) {};

\draw (c) edge node {} (f);
\draw (e1) edge node {} (f);

\draw (1.4,2.3) node[whitenode] (h) {$4^+$};
\draw (g) edge node {} (h);
\draw (e1) edge node {} (h);

\draw (g) edge node {} (2,0.3);
\draw (c) edge node {} (0.7,-0.3);
\draw (1.3,0.2) node[texte] (i) {$5^+$};

\end{tikzpicture}
\label{fig:p2}
}
\caption{Examples of supports.}
\label{fig:parent}
\end{figure}
\captionsetup[subfloat]{labelformat=parens}

\section{Forbidden Configurations}\label{sect:conf}

A \emph{constraint} of an element $u \in V \cup E$ is an already colored element of $V \cup E$ that is adjacent or incident to $u$.  \newline

We define configurations \textbf{($C_1$)} to \textbf{($C_7$)} (see Figure~\ref{fig:config}). Configurations \textbf{($C_1$)}, \textbf{($C_4$)} and \textbf{($C_7$)} are standard. Configurations \textbf{($C_2$)} and \textbf{($C_3$)} follow from the theorem statement. Configuration \textbf{($C_5$)} appears in~\cite{cks07}, and we introduce Configuration \textbf{($C_6$)}.

\begin{itemize}
\item \textbf{($C_1$)} is an edge $(u,v)$ with $d(u)+d(v)\leq k+1$ and $d(u) \leq \lfloor \frac{k}{2} \rfloor$. 
\item \textbf{($C_2$)} is a cycle $(u,v,w,x)$ such that $(u,w)$ is a chord. 
\item \textbf{($C_3$)} is a cycle $(u,v,w,x,y)$ such that $(w,y)$ is a chord. 
\item \textbf{($C_4$)} is a cycle $(u_1,v_1,...,u_p,v_p,u_1), p \geq 2$ where $\forall i, d(v_i)=2$. 
\item \textbf{($C_5$)} is a vertex $v_1$ with $d(v_1)=2$ such that, for $u$ and $x_1$ its two neighbors, there is a path $(v_1,x_1,v_2,\ldots,v_p,x_p,v_{p+1})$ ($p \geq 1$) such that $\forall i, v_i$ is adjacent to $u$, with $\forall \  2 \leq i\leq p, d(v_i)=3$, and $d(v_{p+1})=2$.
\item \textbf{($C_6$)} is a vertex $v_1$ with $d(v_1)=2$ such that, for $u$ and $x_1$ its two neighbors, there is a cycle $(x_1,v_2,x_2,\ldots,x_{p-1},v_p)$ such that $\forall i, v_i$ is adjacent to $u$, and $\forall i \geq 2, d(v_i)=3$.
\item \textbf{($C_7$)} is a vertex $u$ with $d(u)=4$ that has at least two neighbors $u_1$ and $u_2$ with $d(u_1)=d(u_2)=4$.
\end{itemize}

\captionsetup[subfloat]{labelformat=empty}
\begin{figure}[!h]
\centering
\subfloat[][\textbf{($C_1$)}]{
\centering
 \begin{tikzpicture}[scale=0.95]
    \tikzstyle{whitenode}=[draw,circle,fill=white,minimum size=8pt,inner sep=0pt]
    \tikzstyle{blacknode}=[draw,circle,fill=black,minimum size=6pt,inner sep=0pt]
\tikzstyle{tnode}=[draw,ellipse,fill=white,minimum size=8pt,inner sep=0pt]
\tikzstyle{texte} =[fill=white, text=black]

 \draw (-1,4.2) node[whitenode] (u) [label=90:$u$] {}
-- ++(0:1cm) node[whitenode] (v) [label=90:$v$] {};

\draw (-1.8,3.4) node[anchor=text] (a) {$d(u)+d(v)\leq k+1$};
\draw (-1.2,2.9) node[anchor=text] (a) {$d(u) \leq \lfloor \frac{k}{2} \rfloor$};
\end{tikzpicture}
\label{fig:cc1}
}
\qquad
\subfloat[][\textbf{($C_2$)}]{
\centering
 \begin{tikzpicture}[scale=0.95]
    \tikzstyle{whitenode}=[draw,circle,fill=white,minimum size=8pt,inner sep=0pt]
    \tikzstyle{blacknode}=[draw,circle,fill=black,minimum size=6pt,inner sep=0pt]
\tikzstyle{tnode}=[draw,ellipse,fill=white,minimum size=8pt,inner sep=0pt]
\tikzstyle{texte} =[fill=white, text=black]

 \draw (2.3,4.5) node[whitenode] (u) [label=left:$u$] {}
-- ++(-90:1cm) node[whitenode] (v) [label=left:$v$] {}
-- ++(0:1cm) node[whitenode] (w) [label=right:$w$] {}
-- ++(90:1cm) node[whitenode] (x) [label=right:$x$] {};

\draw (u) edge node  {} (x);
\draw (u) edge node  {} (w);
\end{tikzpicture}
\label{fig:cc2}
}
\qquad
\subfloat[][\textbf{($C_3$)}]{
\centering
 \begin{tikzpicture}[scale=0.95]
    \tikzstyle{whitenode}=[draw,circle,fill=white,minimum size=8pt,inner sep=0pt]
    \tikzstyle{blacknode}=[draw,circle,fill=black,minimum size=6pt,inner sep=0pt]
\tikzstyle{tnode}=[draw,ellipse,fill=white,minimum size=8pt,inner sep=0pt]
\tikzstyle{texte} =[fill=white, text=black]

 \draw (5,4.5) node[whitenode] (u) [label=left:$u$] {}
-- ++(-90:1cm) node[whitenode] (v) [label=left:$v$]  {}
-- ++(0:1cm) node[whitenode] (w) [label=right:$w$] {}
-- ++(30:1cm) node[whitenode] (x) [label=right:$x$] {}
-- ++(150:1cm) node[whitenode] (y) [label=right:$y$] {};

\draw (u) edge node  {} (y);
\draw (w) edge node  {} (y);
\end{tikzpicture}
\label{fig:cc3}
}
\qquad
\subfloat[][\textbf{($C_4$)}]{
\centering
 \begin{tikzpicture}[scale=0.95]
    \tikzstyle{whitenode}=[draw,circle,fill=white,minimum size=8pt,inner sep=0pt]
    \tikzstyle{blacknode}=[draw,circle,fill=black,minimum size=6pt,inner sep=0pt]
\tikzstyle{tnode}=[draw,ellipse,fill=white,minimum size=8pt,inner sep=0pt]
\tikzstyle{texte} =[fill=white, text=black]

 \draw (-1,0) node[whitenode] (a1) [label=180:$u_p$] {};
 \draw (0,1) node[blacknode] (b2) [label=90:$v_1$] {};
 \draw (1,0) node[whitenode] (a4) [label=0:$u_3$] {};

\draw (-0.866,0.5) node[blacknode] (b1) [label=180-30:$v_p$] {};
\draw (-0.5,0.866) node[whitenode] (a2) [label=90+30:$u_1$] {};
\draw (0.5,0.866) node[whitenode] (a3) [label=60:$u_2$] {};
\draw (0.866,0.5) node[blacknode] (b3) [label=30:$v_2$] {};

\draw (a1) edge node {} (b1);
\draw (b1) edge node {} (a2);
\draw (a2) edge node {} (b2);
\draw (b2) edge node {} (a3);
\draw (a3) edge node {} (b3);
\draw (b3) edge node {} (a4);

\draw (a1) edge[thick,dotted,bend right] node {} (a4);
\end{tikzpicture}
\label{fig:cc4}
}
\qquad
\subfloat[][\textbf{($C_5$)}]{
\centering
 \begin{tikzpicture}[scale=0.95]
    \tikzstyle{whitenode}=[draw,circle,fill=white,minimum size=8pt,inner sep=0pt]
    \tikzstyle{blacknode}=[draw,circle,fill=black,minimum size=6pt,inner sep=0pt]
\tikzstyle{tnode}=[draw,ellipse,fill=white,minimum size=8pt,inner sep=0pt]
\tikzstyle{texte} =[fill=white, text=black]

 \draw (2,-2) node[whitenode] (e4) [label=180:$x_p$] {}
-- ++(90:1cm) node[blacknode] (d2) [label=180:$v_{p+1}$] {}
-- ++(90:1cm) node[whitenode] (c) [label=180:$u$] {}
-- ++(90:1cm) node[blacknode] (d1) [label=180:$v_1$] {}
-- ++(90:1cm) node[whitenode] (e1) [label=180:$x_1$] {};

\draw (3,1.732) node[blacknode] (f1) [label=60:$v_2$] {};
\draw (3.732,1) node[whitenode] (e2) [label=30:$x_2$] {};
\draw (3,-1.732) node[blacknode] (f3) [label=-60:$v_p$] {};
\draw (3.732,-1) node[whitenode] (e3) [label=-30:$x_{p-1}$] {};

\draw (e1) edge node {} (f1);
\draw (f1) edge node {} (e2);
\draw (e4) edge node {} (f3);
\draw (f3) edge node {} (e3);
\draw (c) edge node {} (f1);
\draw (c) edge node {} (f3);

\draw (e2) edge[thick,dotted,bend left] node {} (e3);
\end{tikzpicture}
\label{fig:cc5}
}
\qquad
\subfloat[][\textbf{($C_6$)}]{
\centering
 \begin{tikzpicture}[scale=0.95]
    \tikzstyle{whitenode}=[draw,circle,fill=white,minimum size=8pt,inner sep=0pt]
    \tikzstyle{blacknode}=[draw,circle,fill=black,minimum size=6pt,inner sep=0pt]
\tikzstyle{tnode}=[draw,ellipse,fill=white,minimum size=8pt,inner sep=0pt]
\tikzstyle{texte} =[fill=white, text=black]
 \draw (5,0) node[whitenode] (g) [label=90:$u$] {}
-- ++(0:1cm) node[blacknode] (j) [label=90:$v_1$] {}
-- ++(0:1cm) node[whitenode] (h1) [label=0:$x_1$] {};

\draw (6.414,1.414) node[blacknode] (i1) [label=45:$v_p$] {};
\draw (5,2) node[whitenode] (h2) [label=90:$x_{p-1}$] {};
\draw (6.414,-1.414) node[blacknode] (i3) [label=-45:$v_2$] {};
\draw (5,-2) node[whitenode] (h3) [label=90:$x_2$] {};

\draw (g) edge node {} (i1);
\draw (g) edge node {} (i3);
\draw (h2) edge node {} (i1);
\draw (h1) edge node {} (i1);
\draw (h3) edge node {} (i3);
\draw (h1) edge node {} (i3);

\draw (h2) edge[thick,dotted,bend right] node {} (h3);
\end{tikzpicture}
\label{fig:cc6}
}
\qquad
\subfloat[][\textbf{($C_7$)}]{
\centering
 \begin{tikzpicture}[scale=0.95]
    \tikzstyle{whitenode}=[draw,circle,fill=white,minimum size=8pt,inner sep=0pt]
    \tikzstyle{blacknode}=[draw,circle,fill=black,minimum size=6pt,inner sep=0pt]
\tikzstyle{tnode}=[draw,ellipse,fill=white,minimum size=8pt,inner sep=0pt]
\tikzstyle{texte} =[fill=white, text=black]

\draw (8,2) node[whitenode] (v1) {}
-- ++(0:1cm) node[blacknode] (u1) [label=45:$u_1$] {}
-- ++(0:1cm) node[blacknode] (u) [label=45:$u$] {}
-- ++(0:1cm) node[blacknode] (u2) [label=45:$u_2$] {}
-- ++(0:1cm) node[whitenode] (v2) {};

\draw (u1)
-- ++(90:1cm) node[whitenode] (w1) {};

\draw (u1)
-- ++(-90:1cm) node[whitenode] (x1) {};

\draw (u)
-- ++(90:1cm) node[whitenode] (w) {};

\draw (u)
-- ++(-90:1cm) node[whitenode] (x) {};

\draw (u2)
-- ++(90:1cm) node[whitenode] (w2) {};

\draw (u2)
-- ++(-90:1cm) node[whitenode] (x2) {};
\end{tikzpicture}
\label{fig:cc7}
}
\caption{Forbidden configurations.}
\label{fig:config}
\end{figure}
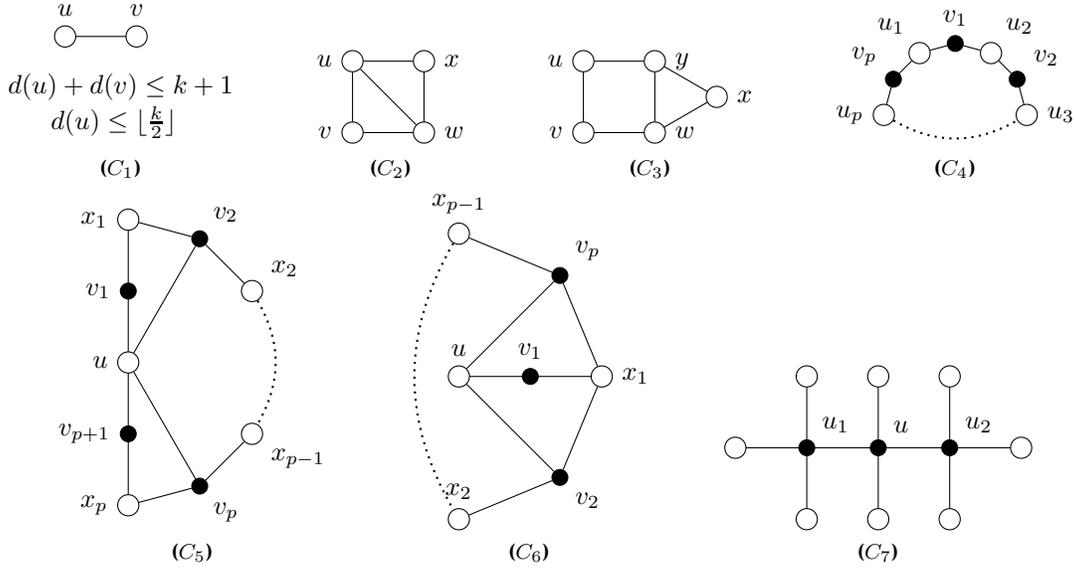
\captionsetup[subfloat]{labelformat=parens}

\begin{lemma}\label{lem:config}
If $G$ is a minimal planar graph such that $\Delta(G) \leq k$, no triangle is adjacent to a cycle of length four, and $\chi''_\ell(G)> k+1$ or $\chi'_\ell(G)>k$, then $G$ cannot contain any of Configurations \textbf{($C_1$)} to \textbf{($C_7$)}.
\end{lemma}

\begin{proof}
First note that in the case of list total coloring, any vertex $v$ in $G$ with $d(v)\leq \frac{k}{2}$ can be colored no matter the coloring of its incident edges and adjacent vertices. Except in the case of \textbf{($C_7$)}, which we deal with appropriately in Claim~\ref{cl:c7}, every edge we have to color is incident to a vertex of degree at most $\frac{k}{2}$, and every vertex we have to color is of degree $\leq \frac{k}{2}$. Consequently, except for \textbf{($C_7$)}, it is safe to consider the problem of list edge coloring only. Indeed, in the case of list total coloring, we can discolor every colored vertex of degree $\leq \frac{k}{2}$, then every edge we have to color has an extra color and at most one extra constraint, and every vertex we have to color will still have enough choices remaining at the end.

\begin{claim}\label{cl:c1}
$G$ cannot contain \textbf{($C_1$)}.
\end{claim}
\begin{proof}
Using the minimality of $G$, we color $G \setminus \{(u,v)\}$. Since $\Delta(G) \leq k$, and $d(u)+d(v) \leq k+1$, the edge $(u,v)$ has at most $k-1$ constraints. There are $k$ colors, so we can color $(u,v)$, thus extending the coloring of $G \setminus \{(u,v)\}$ to $G$.
\end{proof}

\begin{claim}
$G$ cannot contain \textbf{($C_2$)}.
\end{claim}
\begin{proof}
The triangle $(u,v,w)$ shares two edges with the cycle $(u,v,w,x)$ of length $4$.
\end{proof}

\begin{claim}
$G$ cannot contain \textbf{($C_3$)}.
\end{claim}
\begin{proof}
The triangle $(u,x,y)$ shares an edge with the cycle $(u,v,w,x)$ of length $4$.
\end{proof}

\begin{claim}
$G$ cannot contain \textbf{($C_4$)}.
\end{claim}
\begin{proof}
Using the minimality of $G$, we color $G \setminus \{v_i\}_{1 \leq i \leq p}$. Every edge $(u_i,v_i)$ or $(v_i,u_{i+1})$ (subscript taken modulo p) has at most $k-2$ constraints, so there are at least $2$ colors available for each of them. Since even cycles are $2$-choosable, we can color the $(u_i,v_i)$'s and $(v_i,u_{i+1})$'s. Then we can extend the coloring of $G \setminus \{v_i\}_{1 \leq i \leq p}$ to $G$.
\end{proof}

\begin{claim}
$G$ cannot contain \textbf{($C_5$)}.
\end{claim}
\begin{proof}
Let $L : E \rightarrow \mathcal{P}(\mathbb{N})$ be a color assignment such that $\forall a \in E, |L(a)|\geq k$ and such that $G$ is not $L$-colorable. Using the minimality of $G$, we $L$-color $G \setminus \{v_i | 1 \leq i \leq p+1 \}$. We denote by $L'(e)$ the remaining available colors for every edge $e$ that is not colored yet.

Every edge $e$ incident to $u$ and not colored yet has at most $d(u)-(p+1) \leq k-(p+1)$ constraints, thus $|L'(e)|\geq p+1$. Every edge $e$ that is not incident to $u$ and is not colored yet has at most $k-2$ constraints, thus $|L'(e)|\geq 2$. We consider the worst case, i.e. that these inequalities are actually equalities. 

We first consider the case where $L'(v_1,x_1)\not\subset L'(u,v_1)$ or $L'(v_{p+1},x_p)\not\subset L'(u,v_{p+1})$. Consider w.l.o.g. $L'(v_1,x_1)\not\subset L'(u,v_1)$. Color $(v_1,x_1)$ with a color that does not belong to $L'(u,v_1)$, and color arbitrarily $(x_1,v_2),\ldots,(x_p,v_{p+1})$, successively. Then at least $p-1$ colors remain for each $(u,v_i)$ with $2 \leq i \leq p$, while $p$ colors remain for $(u,v_{p+1})$ and $p+1$ for $(u,v_1)$ by assumption. We color arbitrarily $(u,v_2), \ldots, (u,v_{p+1})$, in that order, and finally $(u,v_1)$: then $G$ is $L$-colorable, a contradiction. Thus we can assume from now on that $L'(v_1,x_1)\subset L'(u,v_1)$ and $L'(v_{p+1},x_p)\subset L'(u,v_{p+1})$. We prove the following.
\begin{claimb}\label{clb:rec5}
We can color $\{(u,v_i),(v_i,x_i),(x_i,v_{i+1})| 1 \leq i \leq p\}\setminus \{u,v_1\}$ in such a way that for $L''$ the list assignment of remaining available colors for the edges uncolored yet (here $(u,v_1)$ and $(u,v_{p+1})$), we have $L''(u,v_1) \neq L''(u,v_{p+1})$ if $|L''(u,v_1)|=|L''(u,v_{p+1}|=1$.
\end{claimb}
\begin{proofclaim}
We consider two cases depending on whether $L'(u,v_1)=L'(u,v_{p+1})$.
\begin{itemize}
\item \textit{Assume $L'(u,v_1)\neq L'(u,v_{p+1})$.\\}Let $a$ be a color in $L'(u,v_1)\setminus L'(u,v_{p+1})$. Color $(v_1,x_1)$ with a color other than $a$, then color successively $(x_1,v_2),\ldots,(x_p,v_{p+1}),(u,v_2),\ldots,(u,v_p)$. Now $|L''(u,v_1)|\geq 1, \ |L''(u,v_{p+1})|\geq 1$ and $L''(u,v_1)\neq L''(u,v_{p+1})$ if $|L''(u,v_1)|=|L''(u,v_{p+1})|=1$. Indeed, if $|L''(u,v_{p+1})|=1$ then, since $a \not\in L'(u,v_{p+1})$, the color $a$ does not appear on the edges $(x_1,v_2),\ldots,(x_p,v_{p+1}),(u,v_2),\ldots,(u,v_p)$. Together with the fact that $(v_1,x_1)$ was purposely not colored with $a$ and these are the only uncolored edges around $(u,v_1)$, we have that $a \in L''(u,v_1) \setminus L''(u,v_{p+1})$.
\item \textit{Assume $L'(u,v_1)=L'(u,v_{p+1})$.\\}We color $(v_1,x_1),(x_1,v_2),\ldots,(x_p,v_{p+1})$ as though it were a cycle (i.e. $(v_1,x_1)$ and $(x_p,v_{p+1})$ have to receive different colors): it is possible since even cycles are $2$-choosable. Then we color arbitrarily the $(u,v_i)$'s with $2 \leq i \leq p$. It follows that $L''(u,v_1) \neq L''(u,v_{p+1})$ if $|L''(u,v_1)|=|L''(u,v_{p+1})|=1$. Indeed, $L'(u,v_1)=L'(u,v_{p+1})$, and for $S$ the set of colors on the edges $(u,v_2),\ldots,(u,v_p)$, for $\alpha$ and $\beta$ the colors of $(v_1,x_1)$ and $(v_{p+1})$, we have $L''(u,v_1)=L'(u,v_1)\setminus (S \cup \alpha)$ and $L''(u,v_{p+1})=L'(u,v_1)\setminus (S \cup \beta)$. If $|L''(u,v_1)|=|L''(u,v_{p+1})|=1$, then $\{\alpha,\beta\}\cap S = \emptyset$. Since $\alpha \neq \beta$, this implies $L''(u,v_1) \neq L''(u,v_{p+1})$.
\end{itemize}

\end{proofclaim}

By (\ref{clb:rec5}), we color $\{(u,v_i),(v_i,x_i),(x_i,v_{i+1})| 1 \leq i \leq p\}\setminus \{u,v_1\}$ in such a way that, for $L''$ the list of remaining available colors for $(u,v_1)$ and $(u,v_{p+1})$, we have $L''(u,v_1) \neq L''(u,v_{p+1})$ if $|L''(u,v_1)|=|L''(u,v_{p+1}|=1$. We color arbitrarily $(u,v_1)$ and $(u,v_{p+1})$, starting with the one with fewest available colors if any, thus extending the $L$-coloring of $G \setminus \{v_i | 1 \leq i \leq p+1\}$ to an $L$-coloring of $G$, a contradiction.
\end{proof}

We first prove an intermediary lemma which will be instrumental in the proof of  Claim~\ref{cl:c6}.
\begin{lemma}\label{clb:graphG}
Let $\Gamma$ be $K_{2,3}$, and $y$ (resp. $z$) be a vertex of degree $3$ (resp. $2$) in $\Gamma$. The graph $\Gamma$ is $L_1$-edge-colorable for any list assignment $L_1$ of $3$ colors to each of the two edges incident to $z$ and $2$ colors to each of the other edges, where the two edges incident to $y$ but not to $z$ do not receive the two same colors. 
\end{lemma}
\begin{proof}
We denote $a,b,c$ the three edges incident to $y$, where $c$ is the edge $(y,z)$, and $d$ (resp. $e$, $f$) the other edge incident to $c$ (resp. $b$, $a$). We have $|L_1(a)|=|L_1(b)|=|L_1(e)|=|L_1(f)|=2$ and $|L_1(c)|=|L_1(d)|=3$, with $L_1(a)\neq L_1(b)$.
We consider different cases depending on the list intersections. In the first three cases, we do not use the fact $L(a) \neq L(b)$, which allows us to consider in these cases the problem to be symmetric w.r.t. $(a,b,c)$ and $(f,e,d)$.
\begin{itemize}
\item \textit{Assume $L_1(a)\cap L_1(e)\neq L_1(b)\cap L_1(f)$}.\\W.l.o.g., assume that $(L_1(a)\cap L_1(e))\setminus (L_1(b)\cap L_1(f))\neq \emptyset$ and take an element $\alpha$ of it. Color $a$ and $e$ with $\alpha$. One of $b$ and $f$ still has $2$ colors available. Assume w.l.o.g. it is $b$. Then we color successively $f, d, c$ and $b$.
\item \textit{Assume $L_1(a)\cap L_1(e)=L_1(b)\cap L_1(f)$ and $L_1(a)\cup L_1(e) \neq L_1(b)\cup L_1(f)$}.\\Then assume w.l.o.g. $L_1(a) \subsetneq L_1(b)\cup L_1(f)$. Color $a$ with $\alpha\not\in (L_1(b)\cup L_1(f))$. Then either $L_1(e)=L_1(f)$ and we color $d$ with $\beta \not\in L(f)$, then color successively $c, b, e$ and $f$. Or $L_1(f)\neq L_1(e)$: we color $f$ with a color not in $L_1(e)$, and we can color $(b,c,d,e)$ since even cycles are $2$-choosable.
\item \textit{Assume $L_1(a)\cap L_1(e)=L_1(b)\cap L_1(f)$, $L_1(a)\cup L_1(e) = L_1(b)\cup L_1(f)$ and $L_1(a)\cup L_1(b) \not\subseteq L_1(c)$ or $L_1(e)\cup L_1(f) \not\subseteq L_1(d)$}.\\Then assume w.l.o.g. $L_1(a)\cup L_1(b) \not\subseteq L_1(c)$ and there is $\alpha \in L_1(a)\setminus L_1(c)$. Color $a$ with $\alpha$. If $\alpha \not\in L_1(b)$, color $f, e, b, d$ and $c$, and similarly if $\alpha \not\in L_1(f)$: color $b, e, f, d$ and $c$. If $\alpha \in L_1(b)\cap L_1(f)$, then $\alpha \in L_1(e)$ by assumption. Then we color $e$ with $\alpha$, and color successively $b,f,d$ and $c$.
\item \textit{Assume $L_1(a)\cap L_1(e)=L_1(b)\cap L_1(f)$, $L_1(a)\cup L_1(e) = L_1(b)\cup L_1(f)$, $L_1(a)\cup L_1(b) \subseteq L_1(c)$ and $L_1(e)\cup L_1(f) \subseteq L_1(d)$}.\\Then we must have $L_1(a)=\{1,2\}$, $L_1(b)=\{1,3\}$ and $L_1(c)=\{1,2,3\}$, and for some $\alpha \not\in \{2,3\}$, $L_1(f)=\{\alpha,2\}$, $L_1(e)=\{\alpha,3\}$ and $L_1(d)=\{\alpha,2,3\}$. Then we color $a$ with $1$, $b$ with $3$, $c$ with $2$, $d$ with $3$, $e$ with $\alpha$ and $f$ with $2$. 
\end{itemize}
Thus $\Gamma$ is $L_1$-colorable.
\end{proof}

\begin{claim}\label{cl:c6}
$G$ cannot contain \textbf{($C_6$)}.
\end{claim}
\begin{proof}
Let $L : V \cup E \rightarrow \mathcal{P}(\mathbb{N})$ (resp. $E \rightarrow \mathcal{P}(\mathbb{N})$) be a color assignment such that $\forall a \in V \cup E, |L(a)|\geq k+1$ (resp. $\forall a \in E, |L(a)|\geq k$) and such that $G$ is not $L$-colorable. Using the minimality of $G$, we $L$-color $G \setminus \{v_i | 1 \leq i \leq p\}$. Note that $d(v_i) \leq 3<\frac{k}{2}$ for all $1 \leq i \leq p$, so coloring the edges is enough. We denote by $L'(e)$ the remaining available colors for every edge $e$ that is not colored yet.

Every edge $e$ incident to $u$ and not colored yet has at most $d(u)-p+1$ (resp. $d(u)-p$) constraints, thus $|L'(e)|\geq p$. Every edge $e$ incident to $x_1$ and not colored yet has at most $d(x_1)-3+1$ (resp. $d(x_1)-3$) constraints, thus $|L'(e)|\geq 3$. Every edge $e$ that is not incident to $u$ nor $x_1$ and is not colored yet has at most $k-2+1$ (resp. $k-2$) constraints, thus $|L'(e)|\geq 2$. In the worst case, these inequalities are actually equalities. We first prove the following two claims.
\begin{claimb}\label{clb:rec6}
We can color $\{(u,v_i),(v_i,x_i),(x_i,v_{i+1})| 2 \leq i \leq p-1\}\setminus \{u,v_2\}$ in such a way that for $L''$ the list assignment of remaining available colors for the edges uncolored yet, we have $L''(x_1,v_2) \neq L''(x_1,v_p)$ if $|L''(x_1,v_2)|=|L''(x_1,v_p)|=2$.
\end{claimb}
\begin{proofclaim}
We consider two cases depending on whether $L'(x_1,v_2)=L'(x_1,v_p)$.
\begin{itemize}
\item \textit{Assume $L'(x_1,v_2)\neq L'(x_1,v_p)$.}\\Let $a \in L'(x_1,v_2)\setminus L'(x_1,v_p)$, and color $(v_2,x_2)$ with a color distinct from $a$. Color successively $(x_2,v_3),\ldots,(x_{p-1},v_p)$, then $(u,v_3),\ldots,(u,v_{p-1})$. Now $a\in L''(x_1,v_2)\setminus L''(x_1,v_p)$ unless $|L''(x_1,v_p)|\geq 3$.
\item \textit{Assume $L'(x_1,v_2)=L'(x_1,v_p)$.}\\We color $(v_2,x_2),(x_2,v_3),\ldots,(x_{p-1},v_p)$ as though it were a cycle (i.e. $(v_2,x_2)$ and $(x_{p-1},v_p)$ have to receive different colors): it is possible since even cycles are $2$-choosable. Then we color arbitrarily the $(u,v_i)$'s with $3 \leq i \leq p-1$. It follows that $L''(x_1,v_2) \neq L''(x_1,v_p)$ if $|L''(x_1,v_2)|=|L''(x_1,v_p)|=2$.
\end{itemize}
\end{proofclaim}

By (\ref{clb:rec6}), we color $\{(u,v_i),(v_i,x_i),(x_i,v_{i+1})| 2 \leq i \leq p-1\}\setminus \{u,v_2\}$ in such a way that for $L''$ the list assignment of remaining available colors for the edges uncolored yet, we have $L''(x_1,v_2) \neq L''(x_1,v_p)$ if $|L''(x_1,v_2)|=|L''(x_1,v_p)|=2$. Then, we can assume $|L''(x_1,v_2)|=|L''(x_1,v_p)|=|L''(u,v_2)|=|L''(u,v_p)|=2$, $|L''(u,v_1)|=|L''(v_1,x_1)|=3$ and $L''(x_1,v_2) \neq L''(x_1,v_p)$. Then we color $G$ by Lemma~\ref{clb:graphG}.
\end{proof}

\begin{claim}\label{cl:c7}
$G$ cannot contain \textbf{($C_7$)}.
\end{claim}
\begin{proof}
This is the only situation where list total coloring requires a special argument. The case of list edge coloring is straightforward: we color $G\setminus{(u,u_1)}$, it has at most $6$ adjacent edges, and at least $7$ colors in its list, so we can color it. From now on, we consider specifically the case of list total coloring. Using the minimality of $G$, we color $G\setminus{(u,u_1),(u,u_2)}$. We discolor $u$, $u_1$ and $u_2$. Let $L'$ be the remaining available colors for the edges and vertices that are not colored yet. Since $k \geq 7$ and all the lists are of size at least $k+1$, we have in the worst case $|L'(u_1)|=|L'(u_2)|=2$, $|L'(u,u_1)|=|L'(u,u_2)|=3$ and $|L'(u)|=4$. We consider two cases depending on $L'(u_1)$ and $L'(u,u_2)$.
\begin{itemize}
\item \textit{Assume there exists $a \in L'(u_1) \cap L'(u,u_2)$}. We color $u_1$ and $(u,u_2)$ with $a$, then we color $u_2$, $(u,u_1)$ and $u$.
\item \textit{Assume $L'(u_1) \cap L'(u,u_2)=\emptyset$}. Then $|L'(u_1)\cup L'(u,u_2)|=5> |L'(u)|$. So there exists $a \in L'(u_1)\cup L'(u,u_2) \setminus L'(u)$. Assume $a \in L'(u_1)$ (resp. $a \in L'(u,u_2)$). Color $u_1$ (resp. $(u,u_2)$) with $a$, then color $u_2$, $(u,u_1)$, $(u,u_2)$ (resp. $u_1$) and $u$.
\end{itemize}
Thus the coloring can be extended to $u,u_1,u_2,(u,u_1),(u,u_2)$, a contradiction.
\end{proof}

\end{proof}

\section{Discharging rules}\label{sect:dis}

Given a planar map, we design discharging rules $R_{1.1}$, $R_{1.2}$, $R_{1.3}$, $R_{1.4}$, $R_{2.1}$, $R_{2.2}$, $R_{3.1}$, $R_{3.2}$, $R_4$ and $R_g$ (see Figure~\ref{fig:rules}). We also use a so-called \emph{common pot} which is empty at the beginning, receives weight from some vertices and gives weight to some others.\\
\textbf{\underline{Rules on faces:\newline }}

For any face $f$ of degree at least $4$,

\begin{itemize}
\item Rule $R_1$ is when $f$ is incident to a vertex $u$ of degree $d(u)\leq 3$.
\begin{itemize}
\item Rule $R_{1.1}$ is when $d(f)=4$, and for $v$ the vertex incident to $f$ that is not consecutive to $u$ on the boundary of $f$, we have $d(v)\leq 3$. Then $f$ gives $1$ to $u$.
\item Rule $R_{1.2}$ is when $d(f)=4$, and for $v$ the vertex incident to $f$ that is not consecutive to $u$ on the boundary of $f$, we have $d(v)\geq 4$. Then $f$ gives $\frac{3}{2}$ to $u$.
\item Rule $R_{1.3}$ is when $d(f)\geq 5$ and $d(u)=3$ or $d(u)=2$ and the two neighbors of $u$ are not adjacent. Then $f$ gives $\frac{3}{2}$ to $u$.
\item Rule $R_{1.4}$ is when $d(f)\geq 6$ and $d(u)=2$ such that its two neighbors are adjacent. Then $f$ gives $\frac{5}{2}$ to $u$.
\end{itemize}
\item Rule $R_2$ is when $f$ is incident to a vertex $u$ of degree $4 \leq d(u)\leq 5$.
\begin{itemize}
\item Rule $R_{2.1}$ is when $d(f)=4$ or $d(u)=5$. Then $f$ gives $\frac{1}{2}$ to $u$.
\item Rule $R_{2.2}$ is when $d(f)\geq 5$ and $d(u)=4$. Then $f$ gives $1$ to $u$.
\end{itemize}
\item Rule $R_3$ is when $f$ contains an edge such that there is a vertex $u$ of degree $d(u)=2$ that is adjacent to its two endpoints. Then $f$ gives $\frac{1}{2}$ to $u$.
\end{itemize}
Note that if a vertex $u$ appears more than once on the boundary of $f$, the rules are applied as many times as $u$ appears on the boundary.\newline
\textbf{\underline{Rules on vertices:\newline }}

\begin{itemize}
\item Rule $R_4$ states that for any quadruple $(x,u,u_1,v)$ such that $x$ is the $(v,u_1)$-support of $u$, $x$ gives $\frac{1}{4}$ to $u$. (Note that $R_4$ can be applied twice for the same $x$ and $u$ if there are two different such quadruples involving them).
\item Rule $R_g$ states that for any vertex $x$ of degree $k$, $x$ gives $1$ to the common pot, and every vertex of degree $2$ draws $1$ from it.
\end{itemize}

\captionsetup[subfloat]{labelformat=empty}
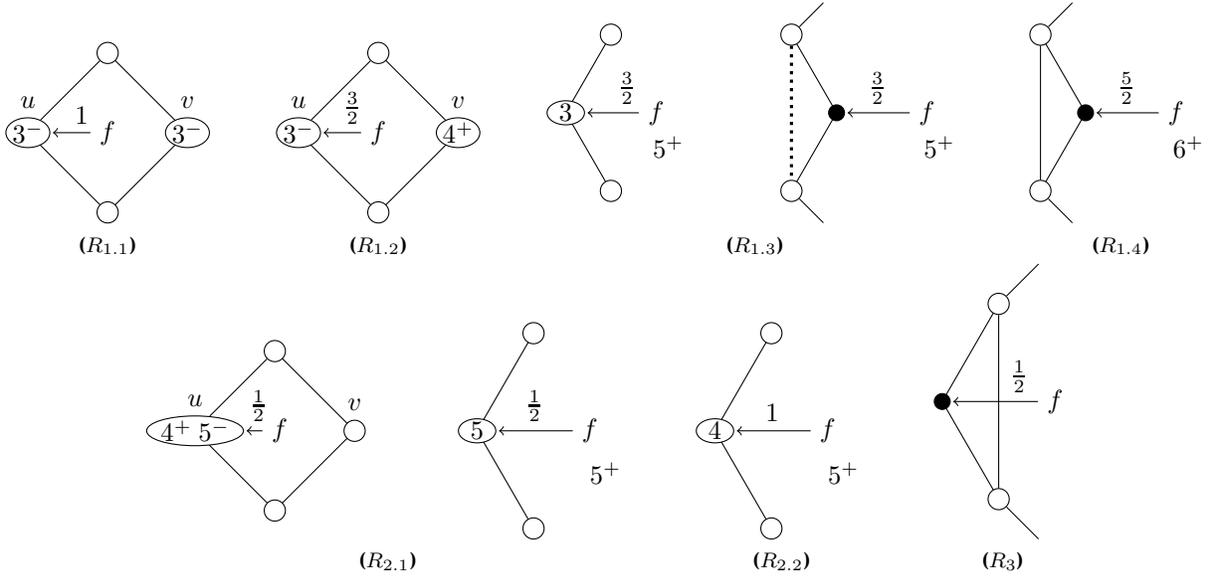
\begin{figure}[!h]
\centering
\subfloat[][\textbf{($R_{1.1}$)}]{
\centering
 \begin{tikzpicture}[scale=1.5,auto]
    \tikzstyle{whitenode}=[draw,circle,fill=white,minimum size=8pt,inner sep=0pt]
    \tikzstyle{blacknode}=[draw,circle,fill=black,minimum size=6pt,inner sep=0pt]
\tikzstyle{tnode}=[draw,ellipse,fill=white,minimum size=8pt,inner sep=0pt]
\tikzstyle{texte} =[fill=white, text=black]

 \draw (0,0) node[tnode] (u) [label=90:$u$] {$3^-$}
-- ++(-45:1cm) node[whitenode] (v) {}
-- ++(45:1cm) node[tnode] (w) [label=90:$v$] {$3^-$}
-- ++(45+90:1cm) node[whitenode] (x) {};

\draw (u) edge node {} (x);
\draw (0.707,0) node[texte] (a) {$f$};

\draw[pre,bend right=0,pos=0.8] (u) edge node {$1$} (a);
\end{tikzpicture}
\label{fig:cc1}
}
\qquad
\subfloat[][\textbf{($R_{1.2}$)}]{
\centering
 \begin{tikzpicture}[scale=1.5,auto]
    \tikzstyle{whitenode}=[draw,circle,fill=white,minimum size=8pt,inner sep=0pt]
    \tikzstyle{blacknode}=[draw,circle,fill=black,minimum size=6pt,inner sep=0pt]
\tikzstyle{tnode}=[draw,ellipse,fill=white,minimum size=8pt,inner sep=0pt]
\tikzstyle{texte} =[fill=white, text=black]

 \draw (0,0) node[tnode] (u) [label=90:$u$] {$3^-$}
-- ++(-45:1cm) node[whitenode] (v) {}
-- ++(45:1cm) node[tnode] (w) [label=90:$v$] {$4^+$}
-- ++(45+90:1cm) node[whitenode] (x) {};

\draw (u) edge node {} (x);
\draw (0.707,0) node[texte] (a) {$f$};

\draw[pre,bend right=0,pos=0.8] (u) edge node {$\frac{3}{2}$} (a);
\end{tikzpicture}
\label{fig:cc2}
}
\qquad
\subfloat[][\textbf{($R_{1.3}$)}]{
\centering
 \begin{tikzpicture}[scale=1.2,auto]
    \tikzstyle{whitenode}=[draw,circle,fill=white,minimum size=8pt,inner sep=0pt]
    \tikzstyle{blacknode}=[draw,circle,fill=black,minimum size=6pt,inner sep=0pt]
\tikzstyle{tnode}=[draw,ellipse,fill=white,minimum size=8pt,inner sep=0pt]
\tikzstyle{texte} =[fill=white, text=black]

 \draw (0,0) node[tnode] (u) {}
-- ++(-90-30:1cm) node[tnode] (v) { $\ 3 \ $}
-- ++(-90+30:1cm) node[tnode] (w) {};

\draw (v)
++ (0:1cm) node[texte] (a) {$f$};
\draw (a)
++ (-70:0.4cm) node[texte] (b) {$5^+$};

\draw[pre,bend right=0,pos=0.8] (v) edge node {$\frac{3}{2}$} (a);

 \draw (2,0) node[tnode] (u) {}
-- ++(-90+30:1cm) node[blacknode] (v) {}
-- ++(-90-30:1cm) node[tnode] (w) {};

\draw[very thick,dotted] (u) edge node {} (w);
\draw (u) edge node {} ++(45:0.5cm);
\draw (w) edge node {} ++(-45:0.5cm);

\draw (v)
++ (0:1cm) node[texte] (a) {$f$};
\draw (a)
++ (-70:0.4cm) node[texte] (b) {$5^+$};

\draw[pre,bend right=0,pos=0.5] (v) edge node {$\frac{3}{2}$} (a);
\end{tikzpicture}
\label{fig:cc3}
}
\qquad
\subfloat[][\textbf{($R_{1.4}$)}]{
\centering
 \begin{tikzpicture}[scale=1.2,auto]
    \tikzstyle{whitenode}=[draw,circle,fill=white,minimum size=8pt,inner sep=0pt]
    \tikzstyle{blacknode}=[draw,circle,fill=black,minimum size=6pt,inner sep=0pt]
\tikzstyle{tnode}=[draw,ellipse,fill=white,minimum size=8pt,inner sep=0pt]
\tikzstyle{texte} =[fill=white, text=black]

 \draw (2,0) node[tnode] (u) {}
-- ++(-90+30:1cm) node[blacknode] (v) {}
-- ++(-90-30:1cm) node[tnode] (w) {};

\draw (u) edge node {} (w);
\draw (u) edge node {} ++(45:0.5cm);
\draw (w) edge node {} ++(-45:0.5cm);

\draw (v)
++ (0:1cm) node[texte] (a) {$f$};
\draw (a)
++ (-70:0.4cm) node[texte] (b) {$6^+$};

\draw[pre,bend right=0,pos=0.5] (v) edge node {$\frac{5}{2}$} (a);
\end{tikzpicture}
\label{fig:cc4}
}
\qquad
\subfloat[][\textbf{($R_{2.1}$)}]{
\centering
 \begin{tikzpicture}[scale=1.5,auto]
    \tikzstyle{whitenode}=[draw,circle,fill=white,minimum size=8pt,inner sep=0pt]
    \tikzstyle{blacknode}=[draw,circle,fill=black,minimum size=6pt,inner sep=0pt]
\tikzstyle{tnode}=[draw,ellipse,fill=white,minimum size=8pt,inner sep=0pt]
\tikzstyle{texte} =[fill=white, text=black]

 \draw (0,0) node[tnode] (u) [label=90:$u$] {$4^+ \ 5^-$}
-- ++(-45:1cm) node[whitenode] (v) {}
-- ++(45:1cm) node[tnode] (w) [label=90:$v$] {}
-- ++(45+90:1cm) node[whitenode] (x) {};

\draw (u) edge node {} (x);
\draw (0.75,0) node[texte] (a) {$f$};

\draw[pre,bend right=0,pos=0.8] (u) edge node {$\frac{1}{2}$} (a);

 \draw (3,0.866) node[tnode] (u) {}
-- ++(-90-30:1cm) node[tnode] (v) { $\ 5 \ $}
-- ++(-90+30:1cm) node[tnode] (w) {};

\draw (v)
++ (0:1cm) node[texte] (a) {$f$};
\draw (a)
++ (-70:0.4cm) node[texte] (b) {$5^+$};

\draw[pre,bend right=0,pos=0.5] (v) edge node {$\frac{1}{2}$} (a);
\end{tikzpicture}
\label{fig:cc5}
}
\qquad
\subfloat[][\textbf{($R_{2.2}$)}]{
\centering
 \begin{tikzpicture}[scale=1.5,auto]
    \tikzstyle{whitenode}=[draw,circle,fill=white,minimum size=8pt,inner sep=0pt]
    \tikzstyle{blacknode}=[draw,circle,fill=black,minimum size=6pt,inner sep=0pt]
\tikzstyle{tnode}=[draw,ellipse,fill=white,minimum size=8pt,inner sep=0pt]
\tikzstyle{texte} =[fill=white, text=black]
 \draw (3,0.866) node[tnode] (u) {}
-- ++(-90-30:1cm) node[tnode] (v) { $\ 4 \ $}
-- ++(-90+30:1cm) node[tnode] (w) {};

\draw (v)
++ (0:1cm) node[texte] (a) {$f$};
\draw (a)
++ (-70:0.4cm) node[texte] (b) {$5^+$};

\draw[pre,bend right=0,pos=0.5] (v) edge node {$1$} (a);
\end{tikzpicture}
\label{fig:cc6}
}
\qquad
\subfloat[][\textbf{($R_3$)}]{
\centering
 \begin{tikzpicture}[scale=1.5,auto]
    \tikzstyle{whitenode}=[draw,circle,fill=white,minimum size=8pt,inner sep=0pt]
    \tikzstyle{blacknode}=[draw,circle,fill=black,minimum size=6pt,inner sep=0pt]
\tikzstyle{tnode}=[draw,ellipse,fill=white,minimum size=8pt,inner sep=0pt]
\tikzstyle{texte} =[fill=white, text=black]

 \draw (2,0) node[tnode] (u) {}
-- ++(-90-30:1cm) node[blacknode] (v) {}
-- ++(-90+30:1cm) node[tnode] (w) {};

\draw[] (u) edge node {} (w);

\draw (v)
++ (0:1cm) node[texte] (a) {$f$};

\draw[pre,bend right=0,pos=0.8] (v) edge node {$\frac{1}{2}$} (a);
\draw (u) edge node {} ++(45:0.5cm);
\draw (w) edge node {} ++(-45:0.5cm);
\end{tikzpicture}
\label{fig:cc7}
}
\caption{Discharging rules $R_1$, $R_2$ and $R_3$.}
\label{fig:rules}
\end{figure}
\captionsetup[subfloat]{labelformat=parens}

\begin{lemma}\label{lem:rules}
A graph $G$ with $\Delta(G) \leq k$ that does not contain Configurations \textbf{($C_1$)} to \textbf{($C_{7}$)} is not planar.
\end{lemma}

\begin{proof}
Assume for contradiction that $G$ is planar. Then it admits an embedding in the plane with no crossing edges. We attribute to each vertex $u$ a weight of $d(u)-6$, and to each face a weight of $2d(f)-6$, and apply discharging rules $R_1$, $R_2$, $R_3$, $R_4$ and $R_g$. 

Since Configurations \textbf{($C_1$)} and \textbf{($C_4$)} do not appear, the subgraph induced in $G$ by the edges incident to a vertex of a degree $2$ is a forest, both its neighbors are of degree $k$. Thus there are at least as many vertices of degree $k$ as there are vertices of degree $2$, so $R_g$ is valid: the common pot does not distribute more weight than it receives.

We first prove the following useful lemma:
\begin{lemma}\label{lem:parent}
In $G$, every vertex $v_0$ with $d(v_0)=2$ that belongs to a face $f_0=(u_1,v_1,u_2,v_0)$ with $d(v_1)=3$ admits a $(v_1,u_1)$-support and a $(v_1,u_2)$-support.
\end{lemma}
\begin{proof}
Assume by contradiction that $v_0$ has no $(v_1,u_1)$-support. Let $(f_0,f_1,\ldots,f_p)$ be a maximal sequence of distinct faces of degree $4$ where $f_i=(u_1,v_{i+1},x_i,v_i)$ (here $x_0=u_2$) and $d(v_{i+1})\leq 3$. Note that $d(v_{i+1})=3$ for every $0 \leq i \leq p$ since Configurations \textbf{($C_5$)} and \textbf{($C_6$)} do not appear. Let $f'$ be the other face to which the edge $(u_1,v_{p+1})$ belongs. We have $d(f')\geq 4$ since $G$ does not contain Configuration \textbf{($C_3$)}. By the contradiction assumption, we have $f'=(u_1,v_{p+1},x_{p+1},v_{p+2})$ with $d(v_{p+2})\leq 3$ as $v_{p+1}$ would otherwise be a $(v_1,u_1)$-support of $v_0$. Since $p$ was chosen to be maximal, we must have $f'=f_0$, a contradiction with the fact that Configuration \textbf{($C_6$)} does not appear in $G$.
\end{proof}

We show that all the vertices have a weight of at least $0$ in the end.\\

Let $u$ be a vertex of $G$. Since Configuration \textbf{($C_1$)} does not appear, $d(u) \geq 2$. We consider different cases depending on the value of $d(u)$.

\begin{enumerate}

\item $d(u)=2$.\\We consider two cases depending on whether $u$ is incident to a triangle.
\begin{enumerate}

\item \textit{Assume $u$ belongs to a triangle $(u,v,w)$}.\\Let $f_1$ and $f_2$ be the two faces adjacent to $(u,v,w)$, where $f_1$ is the face incident to $u$. Then, in order to avoid Configurations \textbf{($C_2$)} and \textbf{($C_3$)}, we must have $d(f_1)\geq 6$ and $d(f_2) \geq 5$. So, by Rules $R_{1.4}$, $R_3$ and $R_g$, $u$ receives $\frac{5}{2}$ from $f_1$, $\frac{1}{2}$ from $f_2$ and $1$ from the common pot. So $u$ has an initial weight of $-4$, gives nothing and receives $4$, so it has a non-negative final weight.

\item \textit{Otherwise, let $f_1$ and $f_2$ be the two faces to which $u$ belongs, with $d(f_1)$, $d(f_2) \geq 4$}.\\For each $f_i \in \{f_1,f_2\}$, we have three cases:
\begin{enumerate}
\item \textit{Either $f_i=(u,u_1,v,u_2)$, with $d(v)\leq 3$}.\\Then $d(v)=3$ since Configuration \textbf{($C_4$)} does not appear, and by Lemma~\ref{lem:parent}, $u$ has a $(v,u_1)$-support, and a $(v,u_2)$-support. Thus, by Rules $R_{1.1}$ and $R_4$, $u$ receives $1$ from $f_i$ and $\frac{1}{4}$ from each of its $(v,\_ )$-supports, so $u$ receives $\frac{3}{2}$ on the side of $f_i$.
\item \textit{Or $f_i=(u,u_1,v,u_2)$, with $d(v)\geq 4$}.\\Then, by Rule $R_{1.2}$, $u$ receives $\frac{3}{2}$ on the side of $f_i$.
\item \textit{Or $d(f_i)\geq 5$}.\\Then, by Rule $R_{1.3}$, $u$ receives $\frac{3}{2}$ on the side of $f_i$.
\end{enumerate}
So $u$ receives $2 \times \frac{3}{2}$ from $f_1$ and $f_2$, and it receives $1$ from the common pot: $u$ has an initial weight of $-4$, gives nothing and receives $4$, so it has a non-negative final weight.
\end{enumerate}

\item $d(u)=3$.\\We consider three cases depending on the faces $u$ is incident to.
\begin{enumerate}

\item \textit{Assume $u$ belongs to a triangle $(u,v,w)$}.\\Let $f_1$ and $f_2$ be the two other faces that are incident to $u$. To avoid Configurations \textbf{($C_2$)} and  \textbf{($C_3$)}, we must have $d(f_1)$, $d(f_2) \geq 5$. So $u$ gives nothing as it cannot be a support. By Rule $R_{1.3}$, $u$ receives $2 \times \frac{3}{2}$, has an initial weight of $-3$ and gives nothing, so it has a non-negative final weight.

\item \textit{Assume $u$ belongs to three faces $f_1=(u,u_1,v_1,u_2)$, $f_2=(u,u_2,v_2,u_3)$ and $f_3=(u,u_3,v_3,u_1)$, with $d(v_1)$, $d(v_2)$, $d(v_3) \leq 3$}.\\Then $u$ cannot be a support so it gives nothing. Vertex $u$ has an initial weight of $-3$, gives nothing, and receives $3 \times 1$ by Rule $R_{1.1}$, so it has a non-negative final weight.

\item \textit{Otherwise, $u$ belongs to a face $f_1$ such that either $d(f_1)\geq 5$ or $f_1=(u,u_1,v_1,u_2)$ with $d(v_1)\geq 4$}.\\Then $u$ has an initial weight of $-3$, gives at most $2 \times \frac{1}{4}$ by $R_4$ as a vertex cannot be support more than twice, and receives at least $\frac{3}{2}+2 \times 1$ by Rule $R_1$, so it has a non-negative final weight.
\end{enumerate}

\item $d(u)=4$.\\We consider two cases depending on whether $u$ is incident to a triangle.
\begin{enumerate}

\item \textit{Assume $u$ is incident to a triangle}.\\Then, since Configurations \textbf{($C_2$)} and  \textbf{($C_3$)} do not appear, $u$ is incident to two faces $f_1$ and $f_2$ such that $d(f_1)$, $d(f_2)\geq 5$. So $u$ has an initial weight of $-2$, gives nothing, and receives at least $2 \times 1$ by Rule $R_{2.2}$, so it has a non-negative final weight.

\item \textit{Otherwise, $u$ is incident to at least $4$ faces of degree at least $4$}.\\Then $u$ has an initial weight of $-2$, gives nothing, and receives at least $4 \times \frac{1}{2}$ by Rule $R_{2}$, so it has a non-negative final weight.
\end{enumerate}

\item $d(u)=5$.\\Since Configurations \textbf{($C_2$)} and  \textbf{($C_3$)} do not appear, $u$ is incident to (at least three and in particular) two faces $f_1$ and $f_2$ such that $d(f_1)$, $d(f_2)\geq 4$. So $u$ has an initial weight of $-1$, gives nothing, and receives at least $2 \times \frac{1}{2}$ by Rule $R_{2.1}$, so it has a non-negative final weight.

\item $6 \leq d(u) \leq k-1$.\\
Vertex $u$ has a non-negative initial weight, gives nothing, receives nothing, so it has a non-negative final weight.

\item $d(u)=k$.\\Then $u$ has an initial weight of at least $1$, gives $1$ to the common pot according to $R_g$ and no other rule applies, so it has a non-negative final weight.
\end{enumerate}

So all the vertices have a non-negative final weight after application of the discharging rules. Let us now prove that the same holds for the faces.\\

Let $f$ be a face of $G$. We consider different cases depending on the value of $d(f)$. Since Configuration \textbf{($C_3$)} does not appear, $f$ cannot give weight according to $R_3$ if $d(f)\leq 4$. Note also that since Configuration \textbf{($C_1$)} does not appear in $G$, $R_3$ can only be applied if the two endpoints of the edge are of degree $k$.

\begin{enumerate}
\item $d(f)=3$.\\Then $f$ has an initial weight of $0$, gives nothing, receives nothing, so it has a non-negative final weight.

\item $d(f)=4$.\\Assume $f=(u,v,w,x)$, where $u$ has the minimum degree. We consider two cases depending on $d(u)$.
\begin{enumerate}

\item $d(u)\leq 3$.\\Then, since Configuration \textbf{($C_1$)} does not appear, $d(v), d(x) \geq 6$, and $f$ gives nothing to them. Face $f$ has an initial weight of $2$. It gives at most $2 \times 1$ to $u$ and $w$ by Rule $R_{1.1}$, or at most $\frac{3}{2}+\frac{1}{2}$ to $u$ and $w$ by Rules $R_{1.2}$ and $R_{2.1}$ (depending on whether $d(w) \leq 3$). So it has a non-negative final weight.

\item $d(u)\geq 4$.\\Then $f$ has an initial weight of $2$, gives at most $4 \times \frac{1}{2}$ to $u$, $v$, $w$ and $x$ by Rule $R_{2.1}$, so it has a non-negative final weight.
\end{enumerate}

\item $d(f)=5$.\\We take $f=(u,v,w,x,y)$, where $u$ has minimum degree, and $d(w) \leq d(x)$. Face $f$ has an initial weight of $4$. We consider different cases depending on $d(u)$.
\begin{enumerate}

\item $d(u)\leq 3$.\\Then, since Configuration \textbf{($C_1$)} does not appear in $G$, $d(v), d(y) \geq 6$. We are in one of the following three cases.
\begin{enumerate}
\item $d(w)\leq 3$.\\Then, since Configuration \textbf{($C_1$)} does not appear in $G$, $d(x) \geq 6$. So, $f$ gives $\frac{3}{2}$ both to $u$ and $w$ by Rule $R_{1.3}$, and may give $\frac{1}{2}$ to a vertex of degree $2$ adjacent to both $x$ and $y$, by Rule $R_3$. So $f$ has an initial weight of $4$, gives at most $\frac{7}{2}$, and has a non-negative final weight. 
\item $4 \leq d(w)\leq 5$.\\Then $f$ gives $\frac{3}{2}$ to $u$ by Rule $R_{1.3}$, at most $1$ to $w$ by Rule $R_2$, and may give $1$ to $x$ by Rule $R_2$ or $\frac{1}{2}$ to a vertex of degree $2$ adjacent to both $x$ and $y$ by Rule $R_3$. So $f$ has an initial weight of $4$, gives at most $\frac{7}{2}$, and has a non-negative final weight. 
\item $d(w) \geq 6$.\\Then $f$ gives $\frac{3}{2}$ to $u$ by Rule $R_{1.3}$, and may give $\frac{1}{2}$ to a vertex of degree $2$ adjacent to both $v$ and $w$, both $w$ and $x$, or both $x$ and $y$, respectively, by Rule $R_3$. So $f$ has an initial weight of $4$, gives at most $\frac{3}{2}+3\times\frac{1}{2}=3$, and has a non-negative final weight. 
\end{enumerate}

\item $d(u)\geq 4$.\\Then, since Configuration \textbf{($C_7$)} does not appear in $G$, there are at most $3$ vertices of degree $4$ in $f$. So $f$ has an initial weight of $4$, gives at most $3\times 1+2\times\frac{1}{2}=4$, by Rules $R_{2.2}$, $R_{2.1}$ and $R_{3}$, and has a non-negative final weight.
\end{enumerate}

\item $d(f)=6$.\\Face $f$ has an initial weight of $6$, so it must not give more than $6$ away. Since Configuration \textbf{($C_3$)} does not appear in $G$, $R_{1.4}$ cannot apply more than once.
We consider four cases depending on the number $N$ of vertices of degree at most $3$ on the boundary of $f$. Note that $N \leq 3$ since Configuration \textbf{($C_1$)} does not appear in $G$. 
\begin{itemize}
\item If $N=0$, then by Rules $R_2$ and $R_3$, $f$ gives at most $d(f)\times 1\leq 6$ away.
\item If $N=1$, then since Configuration \textbf{($C_1$)} does not appear in $G$, $f$ is incident to at most two vertices of degree $4$. Thus, by Rule $R_1$, $f$ gives at most $\frac{5}{2}$ to its only neighbor of degree at most $3$, and by Rules $R_2$ and $R_3$, $f$ gives at most $4\times \frac{1}{2}$ extra weight. So $f$ gives at most $\frac{5}{2}+3 \leq 6$ away.
\item If $N=2$, then since Configuration \textbf{($C_1$)} does not appear in $G$, $f$ is incident to at most one neighbor of degree $4$. Thus, by Rule $R_1$ and since $R_{1.4}$ is applied at most once, $f$ gives at most $\frac{5}{2}+\frac{3}{2}$ to its two incident vertices of degree at most $3$, and by Rules $R_2$ and $R_3$, $f$ gives at most $1$ extra weight. So $f$ gives at most $4+1 \leq 6$ away.
\item Otherwise, $N=3$. Since Configuration \textbf{($C_1$)} does not appear in $G$, $f$ is incident to no vertex of degree $4$, and $R_3$ cannot be applied. Thus, by Rule $R_1$ and since $R_{1.4}$ is applied at most once, $f$ gives at most $\frac{5}{2}+2\times\frac{3}{2}$ to its three incident vertices of degree at most $3$, and neither $R_2$ nor $R_3$ apply. So $f$ gives at most $\frac{11}{2} \leq 6$ away.
\end{itemize}

\item $d(f)\geq 7$.\\In the worst case, $f$ gives $\frac{5}{2} \times \lfloor \frac{d(f)}{2} \rfloor$ by $R_{1.4}$, and it may give an additional $\frac{1}{2}$ by $R_3$ if $d(f)$ is odd, so $f$ has a non-negative final weight. It can easily be checked, as follows.
\begin{enumerate}
\item If $d(f)=7$.\\Then $2d(f)-6-(3\times\frac{5}{2}+\frac{1}{2})=0\geq 0$
\item If $d(f)=8$.\\Then $2d(f)-6-(\frac{d(f)}{2}\times\frac{5}{2})=\frac{3}{4}d(f)-6\geq 0$
\item Otherwise, $d(f)\geq 9$.\\Then $2d(f)-6-(\frac{d(f)}{2}\times\frac{5}{2}+\frac{1}{2})=\frac{3}{4}d(f)-\frac{13}{2}\geq 0$.
\end{enumerate} 
\end{enumerate}

Consequently, after application of the discharging rules, every vertex and every face of $G$ has a non-negative weight, $\sum_{v \in V} (d(v)-6)+ \sum_{f \in F} (2d(f)-6) \geq 0$. Therefore, $G$ is not planar.

\end{proof}

\subsection{Conclusion}

\noindent \emph{Proof of Theorem~\ref{thm:main}}

Let $\Gamma$ be a planar graph with no triangle adjacent to a cycle of length four, such that $\Delta(\Gamma) \geq 7$, and $\Gamma$ is not list edge $\Delta(\Gamma)$-choosable (resp. list total $(\Delta(\Gamma)+1)$-choosable). Graph $\Gamma$ has a subgraph $G$ that is a minimal graph such that $G$ is not list edge $\Delta(\Gamma)$-choosable (resp. list total $(\Delta(\Gamma)+1)$-choosable). We set $k=\Delta(\Gamma) \geq 7$. As $\Delta(G)\leq \Delta(\Gamma)=k$, by Lemma~\ref{lem:config}, graph $G$ cannot contain \textbf{($C_1$)} to \textbf{($C_{7}$)}. Lemma~\ref{lem:rules} implies that $G$ is not planar, thus $\Gamma$ is not planar, a contradiction.
\hfill $\Box$ \vspace{1em}\newline

Consequently, every planar graph with maximum degree $\Delta \geq 7$ and no triangle sharing an edge with a cycle of length four is $\Delta$-edge-choosable and $(\Delta+1)$-total-choosable. If the List Coloring Conjecture is true, then, by a theorem of Sanders and Zhao~\cite{sz01}, it should hold that every planar graph with maximum degree $\Delta \geq 7$ is $\Delta$-edge-choosable, with no condition on the cycles. However, even in the weaker setting where $(\Delta+1)$ colors are allowed, this remains open, and the case $\Delta \geq 8$ has only recently been solved by the first author~\cite{b13}.

\bibliographystyle{plain}

\end{document}